\documentclass[envcountsame]{llncs}
  \usepackage{amsmath}
  \usepackage{amsxtra}
  \usepackage{amssymb}
  \usepackage{latexsym}
\usepackage{hyperref}
\DeclareMathOperator{\Gal}{Gal}
\DeclareMathOperator{\lcm}{lcm}
\DeclareMathOperator{\Tr}{Tr}

\begin{document}
\title{A Function Field Approach Toward Good Polynomials for Further Results on Optimal LRC Codes}
\author{Ruikai ~Chen\inst{1}\and Sihem Mesnager\inst{1,2} }
\institute{Department of Mathematics, University of Paris VIII, F-93526 Saint-Denis, Laboratory Geometry, Analysis and Applications, LAGA, University Sorbonne Paris Nord, CNRS, UMR 7539,  F-93430, Villetaneuse, France,\\
\email{chen\_rk@outlook.com}  \and Telecom Paris, Polytechnic Institute of Paris, 91120 Palaiseau, France.\\
\email{smesnager@univ-paris8.fr}\\} 

\maketitle

\begin{abstract}

Because of the recent applications to distributed storage systems, researchers have introduced a new class of block codes, i.e., locally recoverable (LRC) codes. LRC codes can recover information from erasure(s) by accessing a small number of erasure-free code symbols and increasing the efficiency of repair processes in large-scale distributed storage systems. In this context, Tamo and Barg first gave a breakthrough by cleverly introducing a good polynomial notion. Constructing good polynomials for locally recoverable codes achieving Singleton-type bound (called optimal codes) is challenging and has attracted significant attention in recent years. This article aims to increase our knowledge of good polynomials for optimal LRC codes. Using tools from algebraic function fields and Galois theory, we continue investigating those polynomials and studying them by developing the Galois theoretical approach initiated by Micheli in 2019. Specifically, we push further the study of a crucial parameter $\mathcal G(f)$ (of a given polynomial $f$), which measures how much a polynomial is ``good'' in the sense of LRC codes. We provide some characterizations of polynomials with minimal Galois groups and prove some properties of finite fields where polynomials exist with a specific size of Galois groups. We also present some explicit shapes of polynomials with small Galois groups. For some particular polynomials $f$, we give the exact formula of $\mathcal {G} (f)$.
\end{abstract}



\section{Introduction}

Locally recoverable (LRC) codes can recover information from erasure(s) by accessing a small number of erasure-free code symbols and increasing the efficiency of repair processes in large-scale distributed storage systems. LRC codes and their variants have been extensively studied in recent years. In 2014, Tamo and Barg proposed in a very remarkable paper \cite{tamo2014} a family of locally recoverable codes via so-called \textit{good polynomials}. For an LRC code of locality $r$, a polynomial $f$ over the finite field $\mathbb F_q$ of $q$ elements (where $q$ is a prime power) is called a good polynomial, if
\begin{enumerate}
\item the degree of $f$ is $r+1$;
\item there exists a partition $\{A_1,\dots,A_{\frac{N}{r+1}}\}$ of a set $A\subseteq\mathbb F_q$ of size $N$ into sets of size $r+1$ such that $f$ as a polynomial function is constant on each set $A_i$ in the partition.
\end{enumerate}
A good polynomial is a key ingredient for constructing optimal linear LRC codes. Tamo and Barg also constructed some polynomials with some restrictions. If there is an additive or multiplicative subgroup of order $n$ of $\mathbb F_q$ (i.e., $q\equiv0,1\pmod n$), then the annihilator polynomial of the subgroup is constant on each of its cosets. Based on their work, Liu, Mesnager, and Chen (\cite{liu2018}) presented more general construction approaches using function composition. Also, Liu, Mesnager, and Tang (\cite{liu2020}) have proved that the well-known Dickson polynomials and their composition with some functions are good candidates for good polynomials. Very recently in \cite{Chen-Mesnager-Zhao2021}, good polynomials of low degree over finite fields have been characterized completely, leading to optimal LRC with new flexible localities.

In the sense of coding theory, if there exists $m$ number of such subsets of $\mathbb F_q$ on which $f$ is constant, then one can construct a locally recoverable code of length $mn$. Obviously, those polynomials with $m$ large are preferred. Therefore, it is natural to introduce a parameter indicating how ``good" a polynomial is. For a polynomial $f$ of degree $n$ over $\mathbb F_q$, define 
\[\mathcal G(f)=\left|\{c\in\mathbb F_q\mid f(T)-c\text{ has $n$ distinct roots in $\mathbb F_q$}\}\right|\]
(where $|E|$ denotes the cardinality of a finite set $E$). By a simple investigation we have $\mathcal G(f)\le\lfloor q/n\rfloor$. Micheli (\cite{micheli2020}) discussed this problem in the context of algebraic function fields and Galois theory, pointing out that $\mathcal G(f)$ can be estimated by the order of its corresponding Galois group. In short, given an extension of rational function fields $\mathbb F_q(x)/\mathbb F_q(t)$ defined by $f(x)=t$ with Galois closure $M$, if some condition is satisfied, then $\mathcal G(f)$ is close to $q/[M:\mathbb F_q(t)]$, with an error term $O(\sqrt q)$. Note that up to the error term, the quantity $\mathcal G(f)$ must be $q/m$ for some divisor $m$ of $n!$. We can then characterize the Galois closure $M$, working on the function fields defined by the polynomial $f$, instead of studying the algebraic structure of $\mathbb F_q$.

The remainder of the paper is organized as follows. In Section \ref{pre}, we provide some background on algebraic function fields and notation used through this study. In Section \ref{properties}, we characterize those polynomials with minimal Galois groups and prove some properties of finite fields where polynomials exist with a specific size of Galois groups. By showing that, we know that expected good polynomials do not exist under several circumstances. In Section \ref{instances}, we present some explicit forms of polynomials with small Galois groups by considering the Dickson polynomials of the first kind and powers of linearized polynomials. For some particular polynomials $f$, the exact formula of $\mathcal G(f)$ is given.

\section{Preliminaries}\label{pre}

Let us first recall some basic concepts on algebraic function fields (see \cite{stichtenoth2009} for details). For an algebraic function field $F/K$ ($F$ is a finite extension of $K(t)$
for some $t$ transcendental over $K$), the full constant field consists of all algebraic elements over $K$ in $F$. For a place of $F/K$, let $\mathcal O_P$ be its valuation ring and $v_P$ its discrete valuation. There is a one-to-one correspondence between the places and the discrete valuations of $F/K$. The triangle inequality for a discrete valuation is
\[v_P(a+b)\ge\min\{v_P(a),v_P(b)\},\]
for all $a,b\in F$, where the equality holds if $v_P(a)\ne v_P(b)$. The place $P$ is called a zero of $a\in F$ if $v_P(a)>0$, and a pole if $v_P(a)<0$. Note that $P$ is the unique maximal ideal of $\mathcal O_P$ and $K\subseteq\mathcal O_P$, so $K$ can be embedded into $\mathcal O_P/P$. The degree of $P$ is defined as $\deg(P)=[\mathcal O_P/P:K]$, and a place of degree one is called a rational place.

Assume that $K$ is the full constant field of $F/K$ and let $F^\prime/K^\prime$, with full constant field $K^\prime$, be a finite separable extension of $F/K$. We say a place $P^\prime$ of $F^\prime/K^\prime$ lies above $P$ (or $P$ lies below $P^\prime$) if $P\subseteq P^\prime$. In this case, there exists a positive integer $e=e(P^\prime\mid P)$ such that $v_{P^\prime}(a)=e\cdot v_P(a)$ for all $a\in F$, called the ramification index of $P^\prime$ over $P$. With $\mathcal O_{P^\prime}/P^\prime$ regarded as an extension of $\mathcal O_P/P$, the relative degree of $P^\prime$ over $P$ is $f(P^\prime\mid P)=[\mathcal O_{P^\prime}/P^\prime:\mathcal O_P/P]$. These two numbers satisfy
\begin{equation}\label{equality}\sum_{P^\prime}e(P^\prime\mid P)f(P^\prime\mid P)=[F^\prime:F],\end{equation}
where the sum is extended over all places $P^\prime$ of $F^\prime/K^\prime$ lying above $P$. If $e(P^\prime\mid P)>1$, then $P^\prime$ is said to be ramified in $F^\prime/F$, and if there exists a ramified place of $F^\prime/K^\prime$ lying above $P$, then $P$ is said to be ramified in $F^\prime/F$. Furthermore, a place $P$ of $F$ is said to split completely in $F^\prime$ if $e(P^\prime\mid P)=f(P^\prime\mid P)=1$ for any place $P^\prime$ of $F^\prime/K^\prime$ lying above $P$. If there is a place of $F^\prime/K^\prime$ lying above a rational place of $F/K$ with relative degree one, then $K^\prime=K$. Particularly, if $F^\prime/F$ is Galois, then the Galois group $\Gal(F^\prime/F)$ acts transitively on the set of places of $F^\prime/K^\prime$ lying above $P$, where each place has the same ramification index and relative degree over $P$. As a consequence, \eqref{equality} becomes
\begin{equation}\label{equality_Galois}r\cdot e(P^\prime\mid P)f(P^\prime\mid P)=[F^\prime:F],\end{equation}
for the number $r$ of places of $F^\prime/K^\prime$ lying above $P$.

If $F=K(t)$ for some $t$ transcendental over $K$, then $F/K=K(t)/K$ is called a rational function field. A place of $K(t)/K$ is either the infinite place (the pole of $t$), or the place corresponding to the localization of $K[t]$ at an irreducible polynomial $p(t)$. The latter is simply denoted by $(p(t))$ if there is no ambiguity. We have special interest in a class of extensions of rational function fields. For two relatively prime polynomials $f_0,f_1$ over $K$, let $x$ be an element in some extension of $K(t)$ satisfying $f_0(x)/f_1(x)=t$. Then $K(x)$, as an extension of $K(t)$, is also a rational function field. The minimal polynomial of $x$ over $K(t)$ is $f_0(T)-tf_1(T)\in K(t)[T]$, so that $[K(x):K(t)]=\max\{\deg(f_0),\deg(f_1)\}$. If $f_1(x)=1$ and $K(x)/K(t)$ is separable, then a place of $K(x)$ is ramified in $K(x)/K(t)$ if and only if it is the infinite place or a zero of $f_0^\prime(x)$, where $f_0^\prime$ denotes the formal derivative of $f_0$. In general, for $c\in K$, if $f_0(T)-c$ is factored into irreducible polynomials in $K[T]$ as $\prod_{i=1}^r\varphi_i(T)^{e_i}$, then exactly the $r$ places $(\varphi_1(x)),\dots,(\varphi_r(x))$ of $K(x)$ lie above the place $(t-c)$ of $K(t)$, with $e((\varphi_i(x))\mid(t-c))=e_i$ and $f((\varphi_i(x))\mid(t-c))=\deg(\varphi_i)$.

For our purpose, the two lemmas are also necessary.

\begin{lemma}[{\cite[Lemma 3.9.5]{stichtenoth2009}}]\label{splitting}
Let $F^\prime/F$ be a finite separable extension of function fields and let $L$ be the Galois closure of $F^\prime/F$. If a place $P$ of $F$ splits completely in $F^\prime/F$, then $P$ also splits completely in $L/F$.
\end{lemma}

\begin{lemma}[{\cite[Lemma 6.8]{lidl1993}}]\label{decomposable}
Let $K$ be an arbitrary field and $t=f(x)\in K[x]\setminus K$ for a polynomial $f$ over $K$, such that $K(x)$ and $K(t)$ are rational function fields. Then every intermediate field of $K(x)/K(t)$ is of the form $K(h(x))$ for some polynomials $g,h$ over $K$ such that $f=g\circ h$.
\end{lemma}

Throughout this paper, we use the notation as follows. Given a finite field $\mathbb F_q$ with characteristic $p$, let $f$ be a polynomial of degree $n<q$ over $\mathbb F_q$. Consider the rational function field $\mathbb F_q(t)$ with $t$ transcendental over $\mathbb F_q$, and its finite extension $\mathbb F_q(x)$ defined by $f(x)=t$. Note that if $f(T)\in\mathbb F_q[T^p]$, then $f(T)-c$ will never have $n$ distinct roots in $\mathbb F_q$ for any $c\in\mathbb F_q$. Therefore, suppose $f(T)\notin\mathbb F_q[T^p]$, so that $\mathbb F_q(x)/\mathbb F_q(t)$ is a finite separable extension. Then the splitting field of $f(T)-t$ over $\mathbb F_q(t)$ is exactly the Galois closure of $\mathbb F_q(x)/\mathbb F_q(t)$, denoted by $M$. In this sense, the Galois group of $f$ is defined to be $\Gal(M/\mathbb F_q(t))$. We also assume that $f$ is monic and $f(0)=0$, since $\mathbb F_q(t)=\mathbb F_q(at+b)$ for any $a,b\in\mathbb F_q$ with $a\ne0$.

The properties of $f$ as a mapping on $\mathbb F_q$ are often associated with the structure of its Galois group. Specifically speaking, for $c\in\mathbb F_q$, if $f(T)-c$ splits into $n$ distinct linear factors in $\mathbb F_q[T]$, then the rational place $(t-c)$ of $\mathbb F_q(t)$ splits completely in $\mathbb F_q(x)$. In 1970, Cohen showed the distribution of $f(T)-c$ with prescribed factorization as $c$ varies over $\mathbb F_q$ (\cite{cohen1970}), and recently Micheli gave a more specific description in \cite{micheli2020}. The main result we need is presented in the following lemma.

\begin{lemma}[\cite{micheli2020}]\label{1}
If the full constant field of $M/\mathbb F_q$ is $\mathbb F_q$, then
\[\mathcal G(f)=\frac q{[M:\mathbb F_q(t)]}+O(\sqrt q).\]
More precisely,
\[\frac{q+1-2g\sqrt q}{[M:\mathbb F_q(t)]}-\frac R2\le\mathcal G(f)\le\frac{q+1+2g\sqrt q}{[M:\mathbb F_q(t)]},\]
where $R$ is the number of ramified places of $M/\mathbb F_q(t)$ of degree one, and $g$ is the genus of $M/\mathbb F_q$, both bounded by a constant independent of $q$.
\end{lemma}

The error term $O(\sqrt q)$ is neglectable when $q$ is sufficiently large. 
Therefore, in what follows, we study the splitting field $M$ of $f(T)-t$ over $\mathbb F_q(t)$, as well as the corresponding Galois group. To construct a good polynomial, we can turn to find a polynomial $f$ with $[M:\mathbb F_q(t)]$ as small as possible, such that $\mathbb F_q$ is the full constant field of $M/\mathbb F_q$. Meanwhile, by showing these polynomials' properties, it is clear that good polynomials do not exist in some cases.

\section{Properties of Polynomials with Certain Size of Galois Groups}\label{properties}

By Lemma \ref{1}, the number $\mathcal G(f)$, which we are interested in, is approximately $q/[M:\mathbb F_q(t)]$. For the Galois group $G$ of $f$, we have $n!\ge|G|=[M:\mathbb F_q(t)]=n[M:\mathbb F_q(x)]\ge n$. In most cases, $|G|$ is close to $n!$, so it is significant to discover the properties of $f$ with $|G|$ small enough. Then we can assert that the desired polynomials do not exist with some given conditions. To begin with, we study the extreme case $|G|=n$, or equivalently $M=\mathbb F_q(x)$.

\begin{proposition}\label{deg_1}
The splitting field of $f(T)-t$ over $\mathbb F_q(t)$ is $\mathbb F_q(x)$, if and only if
\[f(T)=(h(T)+c)^k-c^k,\]
where $n=kp^l$ for $k,l\in\mathbb N$ with $q\equiv p^l\equiv1\pmod k$, $h(T)=\sum_{b\in B}(T-b)$ for an additive subgroup $B$ of order $p^l$ in $\mathbb F_q$ such that $\omega B=B$, $\omega$ a primitive $k$-th root of unity in $\mathbb F_q$, and $c=h(\alpha)$ for some $\alpha\in\mathbb F_q$.
\end{proposition}
\begin{proof}
Suppose $n>2$. Since $x$ is one root of $f(T)-t$, and all other roots, denoted by $g_1(x),\dots,g_{n-1}(x)$, lie in $\mathbb F_q(x)$, one has
\begin{equation}\label{f}(-1)^nxg_1(x)\cdots g_{n-1}(x)=-t=-f(x).\end{equation}
For $i=1,\dots,n-1$, it follows from the equation $f(g_i(x))-f(x)=0$ that $g_i(x)$ is integral over $\mathbb F_q[x]$. The fact that $\mathbb F_q[x]$ is integrally closed then implies $g_i(x)\in\mathbb F_q[x]$, and then $\deg(g_i)=[\mathbb F_q(x):\mathbb F_q(g_i(x))]=1$. Now $g_i(x)=a_ix+b_i$ for some $a_i,b_i\in\mathbb F_q$ with $a_i\ne0$. Note that the Galois group $G=\{\sigma_0,\dots,\sigma_{n-1}\}$ of $\mathbb F_q(x)/\mathbb F_q(t)$ is given by $\sigma_0(x)=x$ and $\sigma_i(x)=g_i(x)$, with $\sigma_i^n=\sigma_0$. By induction it is easily seen that
\[x=\sigma_i^n(x)=a_i^nx+(a_i^{n-1}+\dots+a_i+1)b_i,\] 
so $a_i^n=1$.

Suppose $\gcd(n,q-1)=1$, which means $a_i=1$ for each $i$. For $j=0,\dots,n-1$, we have $\sigma_j(g_i(x))=g_i(g_j(x))=x+b_i+b_j$ is also a root of $f(T)-t$. Since $G$ acts transitively on $\{x,g_1(x),\dots,g_{n-1}(x)\}$, for each $i$ we have $b_i+b_j=0$ for some $j$. It follows that the distinct $n$ elements, $0,b_1,\dots,b_{n-1}$, form an additive subgroup of $\mathbb F_q$. This happens only if $n$ divides $q$. From \eqref{f} it follows that
\[f(x)=x(x-b_1)\cdots(x-b_{n-1}).\]

Now let $\gcd(n,q-1)\ne1$. By the same argument, the $n$ (not necessarily distinct) elements $1,a_1,\dots,a_{n-1}$, where each occurs the same time, form a multiplicative subgroup of $\mathbb F_q$. While they are all $n$-th roots of unity, the subgroup is generated by $\omega$, a primitive $k$-th root of unity for some integer $k$ dividing $n$. If $n=k$, then $G$ is cyclic and the $n$ roots of $f(T)-t$ are given by $g_i(x)=\omega^ix+\frac{\omega^i-1}{\omega-1}b_1$, and according to \eqref{f},
\[f(x)=(-1)^{n-1}\prod_{i=0}^{n-1}\left(\omega^ix+\frac{\omega^i-1}{\omega-1}b_1\right),\]
and then
\[\begin{split}f\left(x-\frac{b_1}{\omega-1}\right)=&(-1)^{n-1}\prod_{i=0}^{n-1}\left(\omega^ix-\frac{b_1}{\omega-1}\right)\\=&(-1)^{n-1}\omega^{1+\dots+n-1}\prod_{i=0}^{n-1}\left(x-\frac{\omega^{-i}b_1}{\omega-1}\right)\\=&x^n-\left(\frac{b_1}{\omega-1}\right)^n.\end{split}\]

If $n>k$, then there are exactly $n/k$ roots of $f(T)-t$ in the form: $x,x+\beta_1,\dots,x+\beta_{n/k-1}$, with $\beta_1,\dots,\beta_{n/k-1}\in\mathbb F_q^*$. Apparently $0,\beta_1,\dots,\beta_{n/k-1}$ form an additive subgroup $B$ of order $n/k$, so $n=kp^l$ for some $l\in\mathbb N$. The corresponding automorphisms of $\mathbb F_q(x)/\mathbb F_q(t)$ also form a subgroup $H$ of $G$. More precisely, it is a normal subgroup. If $\sigma\in G\setminus H$, i.e., $\sigma(x)=\omega^ix+\beta$ for some $\beta\in\mathbb F_q$ and $\omega^i\ne1$, then $(\sigma^k)(x)=\omega^{ik}x+\frac{\omega^{ik
}-1}{\omega^i-1}\beta=x$. The order of $\sigma$ then divides $k$, while $\gcd(k,p^l)=1$, so there is no other conjugate of $H$. Since $H$ is normal, we have $\sigma^{-1}H\sigma=H$. Provided $\sigma^\prime\in H$ with $\sigma^\prime(x)=x+\beta_j$ ($0<j<p^l$), and $\sigma(x)=\omega^{-1}x+\beta$, it follows that $\sigma^{-1}(x)=\omega(x-\beta)$ and $(\sigma^{-1}\sigma^\prime\sigma)(x)=\omega(\omega^{-1}x+\beta_j)=x+\omega\beta_j$. Then $\omega B=B$, which happens if and only if $B$ is a vector space over $\mathbb F_p(\omega)$. Thus $p^l$ is a power of the order of $\mathbb F_p(\omega)$, and $p^l\equiv1\pmod k$. Subsequently we determine the specific form of $f$.

By Lemma \ref{decomposable} and the fundamental theorem of Galois theory, there exists an intermediate field $\mathbb F_q(u)$ of $\mathbb F_q(x)/\mathbb F_q(t)$ such that $t=g(u)\in\mathbb F_q[u]$ and $u=h(x)\in\mathbb F_q[x]$. Moreover, $H=\Gal(\mathbb F_q(x)/\mathbb F_q(u))$ and $\mathbb F_q(u)/\mathbb F_q(t)$ is a Galois extension of degree $k=\deg(g)$. Without loss of generality, let $g$ and $h$ be monic. As discussed before, $h(x)-h(0)=\sum_{b\in B}(x-b)$, and $g(u)-g(0)=(u+c)^k-c^k$ for some $c\in\mathbb F_q$. Hence
\[f(x)=\left(\sum_{b\in B}(x-b)+h(0)+c\right)^k-c^k+g(0).\]
It then suffices to set $g(0)=h(0)=0$ for convenience. If one root of $f(T)-t$ is $\omega x+\beta$ for some $\beta\in\mathbb F_q$, then
\[(h(x)+c)^k-c^k=(h(\omega x+\beta)+c)^k-c^k,\]
which means
\[(h(x)+c)^k=(h(\omega x)+h(\beta)+c)^k=(\omega h(x)+h(\beta)+c)^k.\]
It follows that
\[\omega^i(h(x)+c)=\omega h(x)+h(\beta)+c\]
for some integer $i$. Since $x$ is transcendental over $\mathbb F_q$, one has $\omega^i=\omega$ and $c=(\omega-1)^{-1}h(\beta)=h\left((\omega-1)^{-1}\beta\right)$.

With the above discussion, the converse is obvious.
\qed\end{proof}

Now we have characterized those polynomials with minimal Galois groups. It turns out that they coincide with those constructed in \cite{tamo2014}. The following result is already known, but it can be immediately derived as a consequence of the above proposition (cf. Proposition 3.2 and Theorem 3.3 in \cite{tamo2014}). The converse is also true when $q$ is sufficiently large with $n$ fixed.

\begin{corollary}
For the polynomial $f$ of degree $n$ over $\mathbb F_q$, $\mathcal G(f)=\left\lfloor\frac qn\right\rfloor$ if the condition in Proposition \ref{deg_1} is satisfied.
\end{corollary}

\begin{example}
Let $q=64$ and $n=12=3\times2^2$. Note that $64\equiv2^2\equiv1\pmod3$, and $\omega^3=1$ implies $\omega^4-\omega=0$. Set $f(T)=(T^4-T)^3$, so that $M=\mathbb F_q(x)$ and $\mathcal G(f)=5$.
\end{example}

Observe that all those polynomials in Proposition \ref{deg_1} split completely over $\mathbb F_q$. In fact, we can prove more.

\begin{theorem}\label{factorization}
For some $c\in\mathbb F_q$, if $f(T)-c$ has an irreducible factor of multiplicity 1 in $\mathbb F_q[T]$, then the multiplicity of every irreducible factor of $f(T)-c$ divides $[M:\mathbb F_q(x)]$; if $f(T)-c$ has a root in $\mathbb F_q$, then the degree of every irreducible factor of $f(T)-c$ divides $[M:\mathbb F_q(x)]$.
\end{theorem}
\begin{proof}
Let $\varphi_1(T)$ be an irreducible factor of multiplicity 1 of $f(T)-c$ in $\mathbb F_q[T]$ and assume that there is another irreducible factor $\varphi_2(T)$ of $f(T)-c$. Then $(\varphi_1(x))$ and $(\varphi_2(x))$ are places of $\mathbb F_q(x)/\mathbb F_q$ lying above the rational place $(t-c)$ of $\mathbb F_q(t)/\mathbb F_q$, with $e((\varphi_1(x))\mid(t-c))=1$. Let $P_1$ and $P_2$ be places of $M$ lying above $(\varphi_1(x))$ and $(\varphi_2(x))$ respectively. Since $M/\mathbb F_q(x)$ and $M/\mathbb F_q(t)$ are Galois, it follows from \eqref{equality_Galois} that $e(P_1\mid(\varphi_1(x)))$ divides $[M:\mathbb F_q(x)]$, and
\[\begin{split}&e(P_2\mid(\varphi_2(x)))e((\varphi_2(x))\mid(t-c))\\=&e(P_2\mid(t-c))\\=&e(P_1\mid(t-c))\\=&e(P_1\mid(\varphi_1(x)))e((\varphi_1(x))\mid(t-c))\\=&e(P_1\mid(\varphi_1(x))).\end{split}\]
Hence, the multiplicity of $\varphi_2(T)$, equal to $e(\varphi_2(x)\mid(t-c))$, divides $[M:\mathbb F_q(x)]$. Consider the relative degree, and then the second assertion follows immediately.\qed
\end{proof}

This theorem gives an easy way to determine a lower bound of $[M:\mathbb F_q(x)]$ for some polynomial $f$. Note that $f(0)=0$, and we may suppose that $f(T)/T$ has irreducible factors of degree $d_1,\dots,d_k$. Then $[M:\mathbb F_q(x)]$ is a multiple of $\lcm(d_1,\dots,d_k)$. For instance, if $q=19$ and $f(T)=T(T^2+1)(T^3+2T+1)$, where $T^2+1$ and $T^3+2T+1$ are irreducible in $\mathbb F_{19}[T]$, then $[M:\mathbb F_q(x)]$ is divisible by 6. On the other hand, we also learn that it is more likely to obtain a good polynomial for LRC codes, if choosing a polynomial splitting completely over $\mathbb F_q$.

It has been shown that only a few polynomials satisfy the condition for $\mathbb F_q(x)/\mathbb F_q(t)$ being Galois. Now we investigate those for which $[M:\mathbb F_q(x)]>1$. If this is the case, then there must exist another root $y$ of $f(T)-t$ lying outside $\mathbb F_q(x)$, such that $\mathbb F_q(x)$ is isomorphic to $\mathbb F_q(y)$. We will start with the minimal polynomial of $y$ over $\mathbb F_q(x)$.

\begin{lemma}\label{polynomial}
Let $y\in M\setminus\mathbb F_q(x)$ be a root of $f(T)-t$ with minimal polynomial $T^m+a_{m-1}(x)T^{m-1}+\dots+a_0(x)$ over $\mathbb F_q(x)$. Then $a_i(x)\in\mathbb F_q[x]$, $\deg(a_i)\le m-i$ for $0\le i<m$ and $\deg(a_0)=m$.
\end{lemma}
\begin{proof}
Note that $a_i(x)$ can be written as $a_i(x)=u_i(x)/v_i(x)$ for relatively prime polynomials $u_i$ and $v_i$ over $\mathbb F_q$. Then $y^m+a_{m-1}(x)y^{m-1}+\dots+a_0(x)=0$ is equivalent to
\[v(x)y^m+\frac{u_{m-1}(x)v(x)}{v_{m-1}(x)}y^{m-1}+\dots+\frac{u_0(x)v(x)}{v_0(x)}=0,\]
where $v(x)=v_0(x)\cdots v_{m-1}(x)$. Let $d=\gcd(v,\frac{u_{m-1}v}{v_{m-1}},\dots,\frac{u_0v}{v_0})$, so that
\[F(T)=\frac1{d(T)}\left(v(T)y^m+\frac{u_{m-1}(T)v(T)}{v_{m-1}(T)}y^{m-1}+\dots+\frac{u_0(T)v(T)}{v_0(T)}\right)\]
is a polynomial in $\mathbb F_q[y][T]$. Assume that $F(T)$ is reducible over $\mathbb F_q(y)$. By Gauss's Lemma, it is also reducible over $\mathbb F_q[y]$; that is, $F(T)=F_1(T)F_2(T)$ for some $F_1(T),F_2(T)\in\mathbb F_q[y][T]\setminus\mathbb F_q[y]$, with $F_1(x)=0$ or $F_2(x)=0$. If $F_1(T),F_2(T)\notin\mathbb F_q[T]$, then $y$ is a root of a polynomial of lower degree over $\mathbb F_q(x)$. Hence, we may suppose $F_1(T)\in\mathbb F_q[T]$. According to the definition of $d(T)$, this happens only if $F_1(T)\in\mathbb F_q$. The contradiction shows that $F(T)$ is irreducible over $\mathbb F_q(y)$, with a root $x$. The degree of $F(T)$ is
\[[\mathbb F_q(x,y):\mathbb F_q(y)]=\frac{[\mathbb F_q(x,y):\mathbb F_q(t)]}{[\mathbb F_q(x):\mathbb F_q(t)]}=\frac{[\mathbb F_q(x,y):\mathbb F_q(t)]}{[\mathbb F_q(y):\mathbb F_q(t)]}=[\mathbb F_q(x,y):\mathbb F_q(x)],\]
so $\deg(u_0)-\deg(v_0)+\deg(v)-\deg(d)\le m$.

Let $P$ be a place of $\mathbb F_q(x,y)$ lying above the infinite place $P_\infty$ of $\mathbb F_q(x)$ with $v_P$ the corresponding discrete valuation of $P$, and let $e$ be the ramification index of $P$ over $P_\infty$. Note that $f(x)=f(y)$, i.e.,
\[x^n+\alpha_{n-1}x^{n-1}+\dots+\alpha_0=y^n+\alpha_{n-1}y^{n-1}+\dots+\alpha_0,\]
for some $\alpha_0,\dots,\alpha_{n-1}\in\mathbb F_q$. If $v_P(y)\ge0$, then $v_P(f(y))\ge0$, but $v_P(f(x))=-n\cdot e<0$. As a result, $v_P(f(x))=v_P(f(y))=n\cdot v_P(y)$, by the triangle inequality, and then $v_P(y)=-e$. The same argument applies to the conjugates $y_1,\dots,y_m$ of $y$ with respect to $\mathbb F_q(x)$. It follows that
\[v_P(a_0(x))=v_P(y_1)+\dots+v_P(y_m)=-me,\]
and
\[v_P(a_i(x))\ge v_P(y_1)+\cdots+v_P(y_{m-i})=-(m-i)e,\]
for $i=1,\dots,m-1$, using the formula of the elementary symmetric polynomials. This indicates $\deg(u_0)-\deg(v_0)=m$ and $\deg(u_i)-\deg(v_i)\le m-i$. Recalling that $\deg(u_0)-\deg(v_0)+\deg(v)-\deg(d)\le m$ and $d$ divides $v$, we have $\deg(v)=\deg(d)$. Then $v$ divides $\frac{u_{0}v}{v_0},\dots,\frac{u_{m-1}v}{v_{m-1}}$, and clearly $a_0(x),\dots,a_{m-1}(x)\in\mathbb F_q[x]$. This completes the proof.\qed
\end{proof}

\begin{remark}
The inequality $0\le\deg(a_i)\le m-i$ for $0\le i<m$ is sharp in general. For example, let $f(T)=(T^3+1)^2$ over $\mathbb F_q$, where $\gcd(2,q)=\gcd(3,q)=1$. Then there are elements $x$ and $y$ such that $f(x)=f(y)=t$ and $x^3+y^3+2=0$. The minimal polynomial of $y$ over $\mathbb F_q(x)$ is $T^3+x^3+2$. On the other hand, let $f(T)=T^3+T$ over $\mathbb F_q$, where $q\equiv1\pmod3$. It is easy to verify that if $f(y)=f(x)=t$ and $y\ne x$, then the minimal polynomial of $y$ over $\mathbb F_q(x)$ is $T^2+xT+x^2+1$.
\end{remark}

Let $\tau:\mathbb F_q(t)\rightarrow\mathbb F_q(u)$ be an isomorphism fixing $\mathbb F_q$ with $\tau(t)=u$ for some $u$ transcendental over $\mathbb F_q$, and $M^\prime$ be the splitting field of $f(T)-u$. Then $\tau$ can be extended to an isomorphism from $M$ to $M^\prime$. If $f(x)=f(y)=t$, then $f(\tau(x))=f(\tau(y))=u$, and the minimal polynomial of $\tau(y)$ over $\mathbb F_q(\tau(x))$ is obtained by applying $\tau$ to each coefficients of that of $y$ over $\mathbb F_q(x)$. Thus it suffices to study the function field defined by $f(x)=t^n$. To this end, we introduce the field of formal Laurent series over $\overline{\mathbb F_q}$.

\begin{lemma}[{\cite[Theorem 6.12]{lidl1993}}]\label{series}
Let $\overline{\mathbb F_q}((t))$ be the field of formal Laurent series $\sum_{j\ge j_0}c_j/t^j$ with $c_j\in\overline{\mathbb F_q}$. If $\gcd(n,q)=1$, then there exists $\delta(t)=t+\sum_{j\ge0}c_j/t^j\in\overline{\mathbb F_q}((t))$ such that $f(\delta(t))=f(\delta(\omega t))=t^n$, where $\omega\in\overline{\mathbb F_q}$ is a primitive $n$-th root of unity.
\end{lemma}

\begin{theorem}\label{extension}
If $[M:\mathbb F_q(x)]=m$ and $\gcd(n,q)=1$, then $q^m\equiv1\pmod n$.
\end{theorem}
\begin{proof}
With the notation in the above lemma, it suffices to consider the splitting field of $f(T)-t^n$ over $\mathbb F_q(t^n)$. Now that $\delta(\omega t)$ and $\delta(t)$ are different roots of $f(T)-t^n$, it follows from Lemma \ref{polynomial} that
\[\delta(\omega t)^m+a_{m-1}(\delta(t))\delta(\omega t)^{m-1}+\dots+a_0(\delta(t))=0,\]
for some polynomials $a_0,\dots,a_{m-1}$ over $\mathbb F_q$ with $\deg(a_i)\le m-i$ for $i=0,\dots,m-1$. Let $r_i\in\mathbb F_q$ be the $(m-i)$-th coefficient of $a_i$. Then the coefficient of $t^m$ in the above equation is
\[\omega^m+r_{m-1}\omega^{m-1}+\dots+r_0=0,\]
which means $\omega\in\mathbb F_{q^m}$. Since $\omega$ is a primitive $n$-th root of unity, $n$ must divide $q^m-1$.\qed
\end{proof}

If $\gcd(n,q)=1$ and $q^2\equiv1\pmod n$, such polynomial $f$ with $[M:\mathbb F_q(x)]=2$ does exist, and has a unique form. We will leave it in the next section.

\begin{example}
Consider polynomials of degree $n=5$. Note that $\phi(5)=4$ and $q^4\equiv1\pmod5$. If $q\equiv\pm1\pmod5$, then it is easy to find polynomials such that $[M:\mathbb F_q(x)]=2$, as will be seen. If $q\equiv\pm2\pmod5$, then $[M:\mathbb F_q(x)]\ge4$, since in this case neither $q^2\equiv1\pmod5$ nor $q^3\equiv1\pmod5$.
\end{example}

One may ask what happens if $\gcd(n,q)\ne1$. In fact, under some conditions, we can obtain similar results.

\begin{proposition}
Let $[M:\mathbb F_q(x)]=m$ and $n=kp^l$ for some integer $m,k,l$ with $\gcd(k,q)=\gcd(m,q)=1$. If the Galois group $G=\Gal(M/\mathbb F_q(t))$ has a normal Sylow $p$-subgroup, then $q^m\equiv1\pmod k$.
\end{proposition}
\begin{proof}
Let $P$ be the Sylow $p$-subgroup of $G$ with fixed field $E$, and $H=P\cdot\Gal(M/\mathbb F_q(x))$. Then $H$ is a subgroup of order $mp^l$, as $P\cap\Gal(M/\mathbb F_q(x))$ is trivial. The fixed field $F$ of $H$ is $E\cap\mathbb F_q(x)$, and $[F:\mathbb F_q(t)]=|G|/|H|=k$. By Lemma \ref{decomposable}, there exists some polynomials $g$ of degree $k$ over $\mathbb F_q$ such that $F=\mathbb F_q(u)$ with $g(u)=t$. Now $E/\mathbb F_q(t)$ is Galois by the assumption that $P$ is normal, so the Galois closure of $\mathbb F_q(u)/\mathbb F_q(t)$ is contained in $E$, and its degree over $\mathbb F_q(u)$ divides $[E:\mathbb F_q(u)]=m$. It follows from the above theorems that $q^m\equiv1\pmod k$.\qed
\end{proof}

\section{Instances of Polynomials with Small Galois Groups}\label{instances}

\subsection{Dickson Polynomials}

Let $D_n(T,a)$ be the Dickson polynomial (of the first kind) of degree $n$ for some parameter $a\in\mathbb F_q$, defined as
\[D_n(T,a)=\sum_{i=0}^{\lfloor n/2\rfloor}\frac n{n-i}\binom{n-i}i(-a)^iT^{n-2i}.\]
When $a=0$, it is a monomial. A basic property of this polynomial is that
\[D_n(u+\frac au,a)=u^n+\frac{a^n}{u^n}\]
for an indeterminate $u$. Moreover, if $n=kp^l$ for some integers $k,l$ with $\gcd(k,q)=1$, then $D_n(T,a)=D_k(T,a)^{p^l}$. For this reason we may assume $\gcd(n,q)=1$. Let $\omega\in\mathbb F_{q^2}$ be a primitive $n$-th root of unity. Then
\[D_n(T,a)=T\prod_{i=1}^{(n-1)/2}\left(T^2+a(\omega^i-\omega^{-i})^2\right),\]
if $n$ is odd, and
\[D_n(T,a)-D_n(0,a)=T^2\prod_{i=1}^{n/2-1}\left(T^2+a(\omega^i-\omega^{-i})^2\right),\]
if $n$ is even.

There are many other essential properties of the Dickson polynomials. In the context of this paper, it turns out that they can be characterized in another way. First, we shall discuss further the conclusion of Lemma \ref{series}, and then extend Theorem \ref{extension} in the quadratic case, with the explicit form of the polynomial $f$ (still we always assume that $f$ is monic and $f(0)=0$).

Using the notation in Lemma \ref{series}, we claim that if, in addition, $n>2$ and the term of degree $n-1$ of $f$ vanishes, then $c_0=0$ and $c_1\in\mathbb F_q$ for the root $\delta(t)=t+\sum_{j\ge0}c_j/t^j\in\overline{\mathbb F_q}((t))$ of $f(T)-t^n$. Denote $f(T)=T^n+\gamma T^{n-2}+\cdots$ for $\gamma\in\mathbb F_q$. It can be checked that
\[\delta(t)^n=t^n+nc_0t^{n-1}+\left(\frac{n(n-1)}2c_0^2+nc_1\right)t^{n-2}+\cdots,\]
\[\delta(t)^{n-2}=t^{n-2}+(n-2)c_0t^{n-3}+\cdots,\]
\[\delta(t)^{n-3}=t^{n-3}+\cdots.\]
Since $\gcd(n,q)=1$ and $t^n=f(\delta(t))=\delta(t)^n+\gamma\delta(t)^{n-2}+\cdots$, comparing the coefficients we get $c_0=0$ and $nc_1+\gamma=0$.

\begin{theorem}
If $[M:\mathbb F_q(x)]=2$ and $\gcd(n,q)=1$, then $f(T)=D_n(T+b,a)-D_n(b,a)$ for some $a,b\in\mathbb F_q$.
\end{theorem}
\begin{proof}
Without loss of generality, assume that the coefficient of degree $n-1$ of $f$ is 0. Then there is a root $\delta(t)=t+\sum_{j\ge0}c_j/t^j\in\overline{\mathbb F_q}((t))$ of $f(T)-t^n$ , where $c_0=0$ and $c_1\in\mathbb F_q$. The roots of $f(T)-t^n$ are given by $u_0=\delta(t),u_1=\delta(\omega t),\dots,u_{n-1}=\delta(\omega^{n-1}t)$ for a primitive $n$-th root $\omega$ of unity in $\mathbb F_{q^2}$, as shown in Theorem \ref{extension}. For an integer $i$, if $u_i\in\mathbb F_q(u_0)$, then $u_i=\alpha u_0+\beta$ for some $\alpha,\beta\in\mathbb F_q$ with $\alpha\ne0$, as discussed in Proposition \ref{deg_1}. It follows that $\beta=0$ and $\omega^i=\alpha=\omega^{-i}$, which means $\omega^i=-1$ with $n=2i$.

Now suppose that $u_i\notin\mathbb F_q(u_0)$, and let $u_j$ for some integer $j$ be its conjugate with respect to $\mathbb F_q(u_0)$. For $r_1,r_0,s_2,s_1,s_0\in\mathbb F_q$, denote by
\begin{equation}\label{d_0}T^2+(r_1u_0+r_0)T+s_2u_0^2+s_1u_0+s_0\end{equation}
the minimal polynomial of $u_i$ and $u_j$ over $\mathbb F_q(u_0)$. Computing the coefficients of $\delta(\omega^it)^2+(r_1u_0+r_0)\delta(\omega^it)+s_2u_0^2+s_1u_0+s_0$ in $t$ and the same for $j$ leads to
\begin{equation}\label{d_1}\omega^{2i}+r_1\omega^i+s_2=\omega^{2j}+r_1\omega^j+s_2=0,\end{equation}
\begin{equation}\label{d_2}r_0\omega^i+s_1=r_0\omega^j+s_1=0,\end{equation}
\begin{equation}\label{d_3}2c_1+c_1r_1(\omega^i+\omega^{-i})+2c_1s_2+s_0=2c_1+c_1r_1(\omega^j+\omega^{-j})+2c_1s_2+s_0=0.\end{equation}
It follows immediately from $\eqref{d_2}$ that $r_0=s_1=0$. Assume $c_1r_1\ne0$. Then multiplying \eqref{d_3} by $\omega^i$ or $\omega^j$ yields a quadratic equation with two roots $\omega^i$ and $\omega^j$, as well as \eqref{d_1}. Thus $s_2=1$ and $2c_1+2c_1s_2+s_0=c_1r_1^2$. Note that $\omega^i\cdot\omega^j=s_2=1$ and $\omega^i+\omega^j=-r_1$, so
\[s_0=c_1r_1^2-4c_1=c_1((\omega^i+\omega^j)^2-4)=c_1(\omega^i-\omega^{-i})^2.\]
By \eqref{d_0} we know
\[u_iu_{-i}=s_2u_0^2+s_1u_0+s_0=u_0^2+c_1(\omega^i-\omega^{-i})^2,\]
and meanwhile
\[(-1)^nu_0u_1\cdots u_{n-1}=-t^n=-f(u_0).\]
Consequently, if $n$ is odd, then
\[\begin{split}f(u_0)=&u_0\prod_{i=1}^{n-1}u_i=u_0\prod_{i=1}^\frac{n-1}2u_iu_{-i}\\=&u_0\prod_{i=1}^\frac{n-1}2(u_0^2+c_1(\omega^i-\omega^{-i})^2)\\=&D_n(u_0,c_1).\end{split}\]
If $n$ is even, then
\[\begin{split}f(u_0)=&-u_0\prod_{i=1}^{n-1}u_i=-u_0u_{n/2}\prod_{i=1}^{\frac n2-1}u_iu_{-i}\\=&u_0^2\prod_{i=1}^{\frac n2-1}(u_0^2+c_1(\omega^i-\omega^{-i})^2)\\=&D_n(u_0,c_1)-D_n(0,c_1).\end{split}\]

It remains to discuss the case $c_1r_1\ne0$. Note that $c_1$ does not depend on the choice of $i,j$, but $r_1$ does. If $c_1=0$, then $s_1=s_0=0$ by \eqref{d_2} and \eqref{d_3}, and thus $u_iu_j=s_2u_0^2$. This holds for any $i$ such that $u_i\notin\mathbb F_q(u_0)$, so $f(u_0)$ is a product of monomials in $u_0$, namely, $f(u_0)=D_n(u_0,0)$. Now suppose $c_1\ne0$ and $r_1=0$. Then $s_2=-\omega^{2i}$ and $s_0=-2c_1(1-\omega^{2i})$. The function field $\mathbb F_q(u_0,u_i)/\mathbb F_q(u_0)$ is defined by
\[u_i^2-\omega^{2i}u_0^2-2c_1(1-\omega^{2i})=0.\]
If $s_2=1$, then $\omega^{-i}=\omega^j=-\omega^i$, and $s_0=-4c_1=c_1(\omega^i-\omega^{-i})^2$, and the same result follows. If $s_2\ne1$, then $u_i$ is not the conjugate of $u_{-i}$. Recall that $\omega^i+\omega^j=r_1\ne0$ implies $\omega^i\cdot\omega^j=s_2=1$ for any choice of $i,j$. Therefore, the conjugate of $u_{-i}=\delta(\omega^{-i}t)$ is $\delta(-\omega^it)$. Then by the same argument, $\mathbb F_q(u_0,u_{-i})/\mathbb F_q(u_0)$ is defined by
\[u_{-i}^2-\omega^{-2i}u_0^2-2c_1(1-\omega^{-2i})=0.\]
If $q$ is odd, then $\mathbb F_q(u_0,u_i)\ne\mathbb F_q(u_0,u_{-i})$ (to see this, compare their discriminants), which gives rise to a contradiction. If $q$ is even, then $\mathbb F_q(u_0,u_i)/\mathbb F_q(u_0)$ is inseparable, also a contradiction.\qed
\end{proof}

\begin{theorem}\label{dickson}
Suppose $f(T)=D_n(T,a)-D_n(0,a)$ for $a\in\mathbb F_q^*$, $n>2$ with $\gcd(n,q)=1$. If $q\equiv\pm 1\pmod n$, then $[M:\mathbb F_q(x)]=2$. Moreover, the full constant field of $M/\mathbb F_q$ is $\mathbb F_q$ if and only if $q\equiv\pm1\pmod n$.
\end{theorem}
\begin{proof}
Let $t_1=t^n+\frac{a^n}{t^n}-D_n(0,a)$, and $\omega\in\overline{\mathbb F_q}$ a primitive $n$-th root of unity. Then the $n$ distinct roots of $f(T)-t_1\in\mathbb F_q(t_1)[T]$ are given by $u_i=\omega^it+a(\omega^it)^{-1}$, $i=0,\dots,n-1$. Suppose $q\equiv\pm1\pmod n$. Then for $2\le i<n$, it is viable to write $\omega^i=\alpha_i\omega+\beta_i$ for some $\alpha_i,\beta_i\in\mathbb F_q$. If $q\equiv1\pmod n$, set $\alpha_i=\frac{\omega^{-i}-\omega^i}{\omega^{-1}-\omega}$, and then $\omega^{-i}=\omega^i+\alpha_i(\omega^{-1}-\omega)=\alpha_i\omega^{-1}+\beta_i$. If $q\equiv-1\pmod n$, then $\omega^q=\omega^{-1}$ and $\omega^{-i}=(\alpha_i\omega+\beta_i)^q=\alpha_i\omega^{-1}+\beta_i$. It follows that $\omega^it+a(\omega^it)^{-1}=\alpha_i(\omega t+a(\omega t)^{-1})+\beta_i(t+at^{-1})$, which means the splitting field of $f(T)-t_1$ over $\mathbb F_q(t_1)$ is $\mathbb F_q(u_0,u_1)$.

Next we prove that $[\mathbb F_q(u_0,u_1):\mathbb F_q(u_0)]=2$. One can verify that
\[u_1^2-(\omega+\omega^{-1})u_0u_1+u_0^2+a(\omega-\omega^{-1})^2=0,\]
so $u_1$ is integral over $\mathbb F_q[u_0]$, and a fortiori, over $\mathbb F_{q^2}[u_0]$. If $u_1\in\mathbb F_{q^2}(u_0)$, then $u_1\in\mathbb F_{q^2}[u_0]$, since $\mathbb F_{q^2}[u_0]$ as a UFD is integrally closed. Noting that $t$ is transcendental over $\mathbb F_{q^2}$, we have $u_1=c_1u_0+c_0$ for some $c_1,c_0\in\mathbb F_{q^2}$. Comparing the coefficients yields $\omega=c_1$ and $\omega^{-1}a=c_1a$, so $\omega^2=1$ and $n=2$. Thus $u_1\notin\mathbb F_{q^2}(u_0)$, and consequently $[\mathbb F_q(u_0,u_1):\mathbb F_q(u_0)]=2$. If the full constant field of $\mathbb F_q(u_0,u_1)/\mathbb F_q$ is not $\mathbb F_q$, then it must be $\mathbb F_{q^2}$. In this case $\mathbb F_{q^2}(u_0)\subseteq\mathbb F_q(u_0,u_1)$ and $[\mathbb F_{q^2}(u_0):\mathbb F_q(u_0)]=2$. Hence $\mathbb F_q(u_0,u_1)=\mathbb F_{q^2}(u_0)$, but $u_1\notin\mathbb F_{q^2}(u_0)$, a contradiction. Then the full constant field is $\mathbb F_q$.

Suppose now that the full constant field of $M/\mathbb F_q$ is $\mathbb F_q$. Then $\omega+\omega^{-1}=(u_1+u_{n-1})/u_0\in\mathbb F_q$, as it is an algebraic element over $\mathbb F_q$ in $M$. The polynomial $T^2-(\omega+\omega^{-1})T+1$ in $\mathbb F_q[T]$ has roots $\omega$ and $\omega^{-1}$, so we have either $\omega^q=\omega$ or $\omega^q=\omega^{-1}$. This implies $q\equiv1\pmod n$ or $q\equiv-1\pmod n$.\qed
\end{proof}

\begin{remark}
Note that in \cite{Cohen-Matthews1996}, Cohen and Matthews have discussed the monodromy groups of Dickson polynomials. Notably, they showed that Dickson polynomials have a dihedral group as monodromy group in some cases. However, the converse was not considered and our results cannot be derived from \cite{Cohen-Matthews1996}. Furthermore, note that a small mistake has been incorporated in \cite{Cohen-Matthews1996} since $q^2\equiv1\pmod n$ is not equivalent to $q\equiv\pm1\pmod n$ when $n$ is not a prime.
\end{remark}

\begin{remark}
Let $f(T)=D_n(T,a)$ for some $a\in\mathbb F_q^*$. By Theorem \ref{factorization}, whenever there is a root $\alpha\in\mathbb F_q$ of $f(T)-c$ for $c\in\mathbb F_q$, every irreducible factor of $f(T)-c$ has degree 1 or 2. Furthermore, if $n$ is odd, then $f(T)-c$ is either product of linear factors, or product of $T-\alpha$ and quadratic irreducible polynomials over $\mathbb F_q$. To see this, note that for the $n$ roots $x,y_1,\dots,y_{n-1}$ of $f(T)-t$, none of them except $x$ belongs to $\mathbb F_q(x)$. If $\beta\in\mathbb F_q$ is another root of $f(T)-c$, then all other roots of $f(T)-c$ belong to $\mathbb F_q(\alpha,\beta)=\mathbb F_q$. If $n$ is even, it is not hard to draw a similar conclusion that $f(T)-c$ is product of $(T-\alpha)(T+\alpha)$ and quadratic irreducible polynomials over $\mathbb F_q$.
\end{remark}

Now that it has been shown that $[M:\mathbb F_q(x)]=2$ for Dickson polynomials, it is natural to ask what the exact value of $\mathcal G(f)$ is. Before that, we need a basic fact about squares in a finite field.

\begin{lemma}\label{square}
For fixed $c\in\mathbb F_q^*$, the number of $b\in\mathbb F_q$ such that $b^2+c$ is a square in $\mathbb F_q^*$ is $\frac{q-3}2$ if $-c$ is a square in $\mathbb F_q^*$, and $\frac{q-1}2$ otherwise.
\end{lemma}

\begin{theorem}
If $f(T)=D_n(T,a)-D_n(0,a)$ for $a\in\mathbb F_q^*$, $n>2$ with $\gcd(n,q)=1$ and $q\equiv\pm1\pmod n$, then
\[\mathcal G(f)=\begin{cases}\left\lfloor\frac{q-3}{2n}\right\rfloor&\text{if $q$ is odd and }q\equiv\eta(a)\equiv1\pmod n,\\\left\lfloor\frac{q+1}{2n}\right\rfloor&\text{if $q$ is odd and }q\equiv\eta(a)\equiv-1\pmod n,\\\left\lfloor\frac q{2n}\right\rfloor&\text{otherwise},\end{cases}\]
where $\eta$ is the quadratic character of $\mathbb F_q^*$.
\end{theorem}
\begin{proof}
Let $\mathbb F_q(x,y)$ be the function field defined by $y^2-(\omega+\omega^{-1})xy+x^2+a(\omega-\omega^{-1})^2=0$. It is indeed the splitting field of $f(T)-t\in\mathbb F_q(t)[T]$. If a rational place of $\mathbb F_q(t)$ splits completely in $\mathbb F_q(x)$, then it also splits completely in $\mathbb F_q(x,y)$ by Lemma \ref{splitting}, so there are $n$ rational places (other than the infinite one) of $\mathbb F_q(x)$ splitting completely in $\mathbb F_q(x,y)$. Conversely, suppose that $(x-b)$, a place of $\mathbb F_q(x)$ for some $b\in\mathbb F_q$, splits completely in $\mathbb F_q(x,y)$. If $(x-b)$ is unramified in $\mathbb F_q(x)/\mathbb F_q(t)$, then the place $P$ of $\mathbb F_q(t)$ lying below $(x-b)$ splits completely in $\mathbb F_q(x)$, for $\mathbb F_q(x,y)/\mathbb F_q(t)$ is Galois, in which $P$ splits completely. If $(x-b)$ is ramified, then obviously $P$ can not split completely. Therefore it suffices to count the number of $b\in\mathbb F_q$ such that $(x-b)$ is unramified in $\mathbb F_q(x)/\mathbb F_q(t)$ and splitting completely in $\mathbb F_q(x,y)/\mathbb F_q(x)$.

Denote by $\nu$ the number of $b\in\mathbb F_q$ such that $(x-b)$ splits completely in $\mathbb F_q(x,y)/\mathbb F_q(x)$; that is, $T^2-(\omega+\omega^{-1})bT+b^2+a(\omega-\omega^{-1})^2$ has two distinct factors in $\mathbb F_q[T]$. When $q$ is even, that is equivalent to $b\ne0$ and
\[\begin{split}0=&\Tr_{\mathbb F_q/\mathbb F_2}\left(\frac{b^2+a(\omega-\omega^{-1})^2}{(\omega+\omega^{-1})^2b^2}\right)\\=&\Tr_{\mathbb F_q/\mathbb F_2}\left((\omega+\omega^{-1})^{-2}\right)+\Tr_{\mathbb F_q/\mathbb F_2}\left(\frac a{b^2}\right).\end{split}\]
Here we are using the fact that $\Tr_{\mathbb F_q/\mathbb F_2}(\alpha)=0$ for $\alpha\in\mathbb F_q$, if and only if $\beta^2-\beta=\alpha$ for some $\beta\in\mathbb F_q$. Then clearly either $\nu=\frac{q-2}2$, or $\nu=\frac q2$. When $q$ is odd, $(x-b)$ splits completely in $\mathbb F_q(x,y)/\mathbb F_q(x)$ if and only if the quadratic discriminant
\[\begin{split}&(\omega+\omega^{-1})^2b^2-4(b^2+4a(\omega-\omega^{-1})^2)\\=&((\omega+\omega^{-1})^2-4)b^2-4a(\omega-\omega^{-1})^2\\=&(\omega-\omega^{-1})^2(b^2-4a)\end{split}\]
is a square in $\mathbb F_q^*$. Note that if $\omega\notin\mathbb F_q$, then $(\omega-\omega^{-1})^2$ is a non-square in $\mathbb F_q$. In this case the number $\nu$ is obtained from Lemma \ref{square}, as
\[\nu=\begin{cases}\frac{q-3}2&\text{if }q\equiv\eta(a)\equiv1\pmod n,\\\frac{q-1}2&\text{if }q\equiv-\eta(a)\pmod n,\\\frac{q+1}2&\text{if }q\equiv\eta(a)\equiv-1\pmod n.\end{cases}\]

Meanwhile, $(x-b)$ is ramified in $\mathbb F_q(x)/\mathbb F_q(t)$ if and only if $(x-b)$ is a zero of $f^\prime(x)$, i.e. $f^\prime(b)=0$. There are at most $n-1$ such places. Finally we have
\[\frac{\nu-n+1}n\le\mathcal G(f)\le\frac\nu n.\]
There is only one integer in the interval $\left[\frac{\nu-n+1}n,\frac\nu n\right]$. Note that if $q\equiv\pm1\pmod n$, then $\left\lfloor\frac{q-1}{2n}\right\rfloor=\left\lfloor\frac q{2n}\right\rfloor$, and if in addition $q$ is even, then $\left\lfloor\frac{q-2}{2n}\right\rfloor=\left\lfloor\frac q{2n}\right\rfloor$. The desired result then follows.\qed
\end{proof}

\begin{remark}
In \cite{liu2020}, the case $q\equiv1\pmod n$ has been discussed, where the formula is obtained by studying the value sets of Dickson polynomials over finite fields. Here we can generalize it to all cases, with proof using the language of function fields. In particular, if $q\not\equiv1\pmod n$ and $q\not\equiv-1\pmod n$, then $\mathcal G(f)=0$, since the full constant field of $M/\mathbb F_q$ is not $\mathbb F_q$.
\end{remark}

\subsection{Powers of Linearized Polynomials}

In Proposition \ref{deg_1}, it is shown that powers of linearized polynomials are likely to have minimal Galois groups. The following proposition is actually a generalization.

\begin{proposition}
Let $f(T)=\left(h(T)\right)^k$, where $h(T)=\sum_{b\in B}(T-b)$, $n=kp^l$ for some integers $k,l$, with $B$ an additive subgroup of order $p^l$ in $\mathbb F_q$. If $q\equiv1\pmod k$, with $\omega$ a primitive $k$-th root of unity in $\mathbb F_q$, let $j$ be the least positive integer such that $\omega^j=c_0+c_1\omega+\dots+c_{j-1}\omega^{j-1}$ for some $c_i\in\mathbb F_q$ with $c_i B=B$, $0\le i<j$. Then
\begin{enumerate}
\item\label{item1} $[M:\mathbb F_q(x)]\le p^{l(j-1)}$;
\item\label{item2} there exists some $t_0\in\mathbb F_q$ such that $f(T)-t_0$ splits completely (without multiple roots) in $\mathbb F_q[T]$ if and only if there exists $u_0\in\mathbb F_q$ with $u_0\mathbb F_{p^d}\subseteq\{h(\alpha)\mid\alpha\in\mathbb F_q\}$ and $p^{l+d}\le q$, where $d$ is the least positive integer such that $p^d\equiv1\pmod k$;
\item\label{item3} $\frac d{\gcd(d,l)}\le j\le d$, and the lower bound is achieved if $B=\mathbb F_{p^l}$.
\end{enumerate}
\end{proposition}
\begin{proof}
(1) Let $u=h(x)$. For $1\le i<j$, choose an element $y_i$ such that $h(y_i)=\omega^iu$. Clearly $f(y_i+b_0)=u^k=f(x)$ for any $b_0\in B$ and $[\mathbb F_q(x,y_i):\mathbb F_q(x)]\le p^l$. Accordingly $[\mathbb F_q(x,y_1,\dots,y_{j-1}):\mathbb F_q(x)]\le p^{l(j-1)}$. If $cB=B$ for some $c\in\mathbb F_q$, then $B$ is a vector space over $\mathbb F_p(c)$, and thus $c\in\mathbb F_{p^l}$. Let $y_j=c_0x+c_1y_1+\dots+c_{j-1}y_{j-1}$, so that
\[\begin{split}h(y_j)=&h(c_0x)+h(c_1y_1)+\dots+h(c_{j-1}y_{j-1})\\=&c_0^{p^l}h(x)+c_1^{p^l}h(y_1)+\dots+c_{j-1}^{p^l}h(y_{j-1})\\=&\left(c_0+c_1\omega+\dots+c_{j-1}\omega^{j-1}\right)u\\=&\omega^ju.\end{split}\]
Similarly we have $y_{j+1}=c_0y_1+c_1y_2+\dots+c_{j-1}y_j$ with $h(y_{j+1})=\omega^{j+1}u$, and so on. By adding each element of $B$ to $x,y_1,\dots,y_{n-1}$, we obtain $n$ distinct roots of $f(T)-t$. Thus $M=\mathbb F_q(x,y_1,\dots,y_{j-1})$ and $[M:\mathbb F_q(x)]\le p^{l(j-1)}$. This completes the proof.

(2) Note that $f(T)-t_0$ splits completely in $\mathbb F_q[T]$ if and only if $t_0=u_0^k$ for some $u_0\in\mathbb F_q$ and $h(T)-\omega^iu_0$ splits completely for any integer $i$. The latter is equivalent to $\omega^iu_0\in E$, where $E=\{h(\alpha)\mid\alpha\in\mathbb F_q\}$ is a vector space over $\mathbb F_p$. Meanwhile $\mathbb F_{p^d}$ is the smallest subfield of $\mathbb F_q$ containing $\omega$, as well as the smallest vector space over $\mathbb F_p$ in $\mathbb F_q$ containing $1,\omega,\dots,\omega^{k-1}$, for the minimal polynomial of $\omega$ over $\mathbb F_p$ has degree $d$. It follows that $u_0\mathbb F_{p^d}\subseteq E$. Comparing the dimensions we have $p^{l+d}\le q$.

(3) Let $\mathbb F_{p^e}$ for some integer $e$ be the largest subfield of $\mathbb F_q$ over which $B$ is a vector space. Then it is clear that $e$ divides $l$, and $cB=B$ if and only if $c\in\mathbb F_{p^e}$; hence $j$ is the degree of $\omega$ over $\mathbb F_{p^e}$, which is exactly the least positive integer such that $p^{ej}\equiv1\pmod k$. This implies $j=\frac d{\gcd(d,e)}$, and the inequality follows.\qed
\end{proof}

With the notation above, we give an example.

\begin{example}
Let $q=64$, $n=6$, $k=3$ and $B=\{0,c\}$ for some $c\in\mathbb F_q^*$. Since $\omega^2+\omega+1=0$, applying the proposition we have $[M:\mathbb F_q(x)]\le2$, but $\omega B\ne B$, so $M\ne\mathbb F_q(x)$. It turns out that $[M:\mathbb F_q(x)]=2$. On the other hand, $\omega$ lies in $\mathbb F_4$ and $c^2\mathbb F_4\subseteq c^2\mathbb F_{32}=\{c^2\alpha(\alpha-1)\mid\alpha\in\mathbb F_q\}=\{\alpha(\alpha-c)\mid\alpha\in\mathbb F_q\}$, so there exists a rational place of $\mathbb F_q(t)$ splitting completely in $\mathbb F_q(x)$, and hence in $M$. Then the full constant field of $M/\mathbb F_q$ is $\mathbb F_q$. This example shows that without the assumption $\gcd(n,q)=1$, Theorem \ref{extension} is not valid.
\end{example}

\section{Conclusions}

This paper has discussed a property of polynomials over finite fields with applications to locally recoverable codes and turned to characterize the corresponding Galois groups over function fields. For a polynomial of degree $n$ with $[M:\mathbb F_q(x)]$ being a specific integer, some of its properties have been presented in terms of the polynomial factorization and the arithmetic of $q$ and $n$. Besides, there are also some specific forms of polynomials with good properties, especially the Dickson polynomials. In most cases, we also proved that such polynomials are unique. The results may be applied to other research on polynomials over finite fields and their corresponding function fields. However, there are still many polynomials for which it is difficult to give a more precise condition with respect to their Galois groups or splitting fields. This may be an interesting and challenging problem.

\section*{Acknowledgement}

The authors sincerely thank the anonymous referees and the Associate Editor for their constructive and valuable comments, which have improved the quality of the paper highly. The work of the first author is supported by the China Scholarship Council. The funding corresponds to the scholarship for the Ph.D. thesis of the first author in Paris, France.

\bibliographystyle{splncs04}
\bibliography{ref}

\begin{thebibliography}{1}
\providecommand{\url}[1]{\texttt{#1}}
\providecommand{\urlprefix}{URL }
\providecommand{\doi}[1]{https://doi.org/#1}

\bibitem{Chen-Mesnager-Zhao2021}
Chen, R., Mesnager, S., Zhao, C.A.: Good polynomials for optimal LRC of low
  locality. Des. Codes Cryptogr  \textbf{7}(89),  1639--1660 (2021)

\bibitem{cohen1970}
Cohen, S.D.: The distribution of polynomials over finite fields. Acta
  Arithmetica  \textbf{3}(17),  255--271 (1970)

\bibitem{Cohen-Matthews1996}
Cohen, S.D., Matthews, R.W.: Monodromy groups of classical families over finite
  fields. Finite Fields and Applications, London Math Soc. Lecture Note Ser.,
  233, Cambridge Univ. Press, Cambridge (233),  59--68 (1996)

\bibitem{lidl1993}
Lidl, R., Mullen, G.L., Turnwald, G.: Dickson polynomials. Longman,
  London-Harlow-Essex (1993)

\bibitem{liu2018}
Liu, J., Mesnager, S., Chen, L.: New constructions of optimal locally
  recoverable codes via good polynomials. IEEE Transactions on Information
  Theory  \textbf{64}(2),  889--899 (2018)

\bibitem{liu2020}
Liu, J., Mesnager, S., Tang, D.: Constructions of optimal locally recoverable
  codes via dickson polynomials. Designs, codes and cryptography. To appear
  (2020)

\bibitem{micheli2020}
Micheli, G.: Constructions of locally recoverable codes which are optimal. IEEE
  Transactions on Information Theory  \textbf{66}(1),  167--175 (2020)

\bibitem{stichtenoth2009}
Stichtenoth, H.: Algebraic function fields and codes, vol.~254. Springer
  Science \& Business Media (2009)

\bibitem{tamo2014}
Tamo, I., Barg, A.: A family of optimal locally recoverable codes. IEEE
  Transactions on Information Theory  \textbf{60}(8),  4661--4676 (2014)

\end{thebibliography}

\end{document}